\documentclass[11pt]{article}
\usepackage[left=1.5in,top=1.5in,right=1.5in,bottom=1.25in]{geometry}
\usepackage{amsmath}
\usepackage{amssymb}
\usepackage{amsbsy}
\usepackage{amsthm}
\usepackage{epsfig}
\usepackage[round]{natbib}
\usepackage{clrscode}
\usepackage{setspace}
\usepackage{multirow}

\usepackage{mathptmx}

\newcommand{\N}{\mbox{N}}

\newcommand{\Be}{\mbox{Be}}
\newcommand{\dd}{\mbox{d}}
\newcommand{\IG}{\mbox{IG}}

\newcommand{\E}{\mbox{E}}

\newcommand{\bbeta}{\boldsymbol{\beta}}

\newcommand{\by}{\mathbf{y}}
\newcommand{\p}{p}

\newcommand{\kummer}{\phantom{}_1 F_1}
\newcommand{\kummertwo}{\phantom{}_2 F_1}

\newtheorem{theorem}{Theorem}

\newtheorem{proposition}[theorem]{Proposition}

\usepackage{url}

\begin{document}


\title{On the half-Cauchy prior for a global scale parameter}

\author{Nicholas G.~Polson\footnote{nicholas.polson@chicagobooth.edu}  \\
University of Chicago \\
\\
James G.~Scott\footnote{james.scott@mccombs.utexas.edu} \\
The University of Texas at Austin}

\date{First draft: June 2010\\
This draft: September 2011}

\maketitle

\begin{abstract}
This paper argues that the half-Cauchy distribution should replace the inverse-Gamma distribution as a default prior for a top-level scale parameter in Bayesian hierarchical models, at least for cases where a proper prior is necessary.  Our arguments involve a blend of Bayesian and frequentist reasoning, and are intended to complement the original case made by \citet{gelman:2006} in support of the folded-$t$ family of priors.  First, we generalize the half-Cauchy prior to the wider class of hypergeometric inverted-beta priors.  We derive expressions for posterior moments and marginal densities when these priors are used for a top-level normal variance in a Bayesian hierarchical model.  We go on to prove a proposition that, together with the results for moments and marginals, allows us to characterize the frequentist risk of the Bayes estimators under all global-shrinkage priors in the class.  These theoretical results, in turn, allow us to study the frequentist properties of the half-Cauchy prior versus a wide class of alternatives.  The half-Cauchy occupies a sensible ``middle ground'' within this class: it performs very well near the origin, but does not lead to drastic compromises in other parts of the parameter space.  This provides an alternative, classical justification for the repeated, routine use of this prior.  We also consider situations where the underlying mean vector is sparse, where we argue that the usual conjugate choice of an inverse-gamma prior is particularly inappropriate, and can lead to highly distorted posterior inferences.  Finally, we briefly summarize some open issues in the specification of default priors for scale terms in hierarchical models.

\vspace{1pc}


\vspace{1pc}

\noindent Acknowledgements: the authors would like to thank an anonymous referee, the editor, and the associate editor, whose many useful suggestions have helped us to improve the manuscript.
\end{abstract}

\section{Introduction}

Consider a normal hierarchical model where, for $i = 1, \ldots, p$,
\begin{eqnarray*}
(y_{i} \mid \beta_i, \sigma^2) &\sim& \N(\beta_i, \sigma^2) \\
(\beta_i \mid \lambda^2, \sigma^2) &\sim& \N(0, \lambda^2 \sigma^2) \\
\lambda^2 &\sim& g(\lambda^2) \, .
\end{eqnarray*}
This prototype case embodies a very general problem in Bayesian inference: how to choose default priors for top-level variances (here $\lambda^2$ and $\sigma^2$) in a hierarchical model.

The routine use of Jeffreys' prior for the error variance, $ p( \sigma^2 ) \propto \sigma^{-2} $, poses no practical issues.
This is not the case for $ p( \lambda^2 ) $, however, as the improper prior $ p( \lambda^2 ) \propto \lambda^{-2} $ leads
to an improper posterior. This can be seen from the marginal likelihood:
$$
p( y \mid \lambda^2 ) \propto \prod_{i=1}^p ( 1 + \lambda^2 )^{- \frac{1}{2}} \exp \left ( - \frac{1}{2} \sum_{i=1}^p
 \frac{ y_i^2 }{1 + \lambda^2 } \right ) \, ,
$$
where we have taken $ \sigma^2 = 1 $ for convenience.  This is positive at $ \lambda^2 = 0$; therefore, whenever the prior $p(\lambda^{2})$ fails to be integrable at the origin, so too will the posterior.  A number of default choices have been proposed to overcome this issue.  A classic reference is \citet{tiao:tan:1965}; a very recent one is \citet{morris:tang:2011}, who use a flat prior $ p( \lambda^2 ) \propto 1 $.

We focus on a proposal by \citet{gelman:2006}, who studies the class of half-$t$ priors for the scale parameter $\lambda$:
$$
p(\lambda \mid d) \propto \left( 1 + \frac{\lambda^2}{d} \right)^{-(d+1)/2} \, 
$$
for some degrees-of-freedom parameter $d$.  The half-$t$ prior has the appealing property that its density evaluates to a nonzero constant at $\lambda = 0$. This distinguishes it from the usual conjugate choice of an inverse-gamma prior for $\lambda^2$, whose density vanishes at $\lambda = 0$.  As \citet{gelman:2006} points out, posterior inference under these priors is no more difficult than it is under an inverse-gamma prior, using the simple trick of parameter expansion.

These facts lead to a simple, compelling argument against the use of the inverse-gamma prior for variance terms in models such as that above.  Since the marginal likelihood of the data, considered as a function of $\lambda$, does not vanish when $\lambda = 0$, neither should the prior density $\p(\lambda)$.  Otherwise, the posterior distribution for $\lambda$ will be inappropriately biased away from zero.  This bias, moreover, is most severe near the origin, precisely in the region of parameter space where the benefits of shrinkage become most pronounced.

This paper studies the special case of a half-Cauchy prior for $\lambda$ with three goals in mind.  First, we embed it in the wider class of hypergeometric inverted-beta priors for $\lambda^2$, and derive expressions for the resulting posterior moments and marginal densities.  Second, we derive expressions for the classical risk of Bayes estimators arising from this class of priors.  In particular, we prove a result that allows us to characterize the improvements in risk near the origin ($\Vert \bbeta \Vert \approx 0$) that are possible using the wider class.  Having proven our risk results for all members of this wider class, we then return to the special case of the half-Cauchy; we find that the frequentist risk profile of the resulting Bayes estimator is quite favorable, and rather similar to that of the positive-part James--Stein estimator.  Therefore Bayesians can be comfortable using the prior on purely frequentist grounds.

Third, we attempt to provide some insight about the use of such priors in situations where $\bbeta$ is expected to be sparse.  We find that the arguments of \citet{gelman:2006} in favor of the half-Cauchy are, if anything, amplified in the presence of sparsity, and that the inverse-gamma prior can have an especially distorting effect on posterior inference for sparse signals. 

Overall, our results provide a complementary set of arguments in addition to those of \citet{gelman:2006} that support the routine use of the half-Cauchy prior: its excellent (frequentist) risk properties, and its sensible behavior in the presence of sparsity compared to the usual conjugate alternative.  Bringing all these arguments together, we contend that the half-Cauchy prior is a sensible default choice for a top-level variance in Gaussian hierarchical models.  We echo the call for it to replace inverse-gamma priors in routine use, particularly given the availability of a simple parameter-expanded Gibbs sampler for posterior computation.

\section{Inverted-beta priors and their generalizations}

Consider the family of inverted-beta priors for $\lambda^2$:
$$
p(\lambda^2) = \frac{ (\lambda^2)^{b-1} \ (1+\lambda^2)^{-(a+b)}} {\Be(a,b)} \, ,
$$
where $\Be(a,b)$ denotes the beta function, and where $a$ and $b$ are positive reals.  A half-Cauchy prior for $\lambda$ corresponds to an inverted-beta prior for $\lambda^2$ with $a = b = 1/2$.  This family also generalizes the robust priors of \citet{strawderman:1971} and \citet{BergerAnnals1980}; the normal-exponential-gamma prior of \citet{griffin:brown:2005}; and the horseshoe prior of \citet{Carvalho:Polson:Scott:2008a}.  The inverted-beta distribution is also known as the beta-prime or Pearson Type VI distribution.  An inverted-beta random variable is equal in distribution to the ratio of two gamma-distributed random variables having shape parameters $a$ and $b$, respectively, along with a common scale parameter.

The inverted-beta family is itself a special case of a new, wider class of hypergeometric inverted-beta distributions having the following probability density function:
\begin{equation}
\label{MTIBdensity}
p(\lambda^2) = C^{-1} (\lambda^2)^{b-1} \ (\lambda^2 + 1)^{-(a+b)} \ \exp \left\{ -\frac{s}{1+\lambda^2} \right\} \ \left\{ \tau^2 + \frac{1-\tau^2}{1+\lambda^2} \right\}^{-1} \, ,
\end{equation}
for $a>0$, $b>0$, $\tau^2>0$, and $s \in \mathcal{R}$.  This comprises a wide class of priors leading to posterior moments and marginals that can be expressed using confluent hypergeometric functions.  In Appendix \ref{hyperg.integral.details.appendix} we give details of these computations, which yield
\begin{equation}
\label{normalizing.constant.phi1}
C = e^{-s} \ \Be(a, b) \ \Phi_1(b, 1, a + b, s, 1-1/\tau^2) \, ,
\end{equation}
where $\Phi_1$ is the degenerate hypergeometric function of two variables \citep[9.261]{gradshteyn:ryzhik:1965}.  This function can be calculated accurately and rapidly by transforming it into a convergent series of $\kummertwo$ functions  \citep[\S 9.2 of][]{gradshteyn:ryzhik:1965, gordy:1998}, making evaluation of (\ref{normalizing.constant.phi1}) quite fast for most choices of the parameters.

Both $\tau$ and $s$ are global scale parameters, and do not control the behavior of $\p(\lambda)$ at $0$ or $\infty$.  The parameters $a$ and $b$ are analogous to those of beta distribution.  Smaller values of $a$ encourage heavier tails in $\p(\beta$), with $a=1/2$, for example, yielding Cauchy-like tails.  Smaller values of $b$ encourage $\p(\beta)$ to have more mass near the origin, and eventually to become unbounded; $b = 1/2$ yields, for example, $\p(\beta) \approx \log(1+1/\beta^2)$ near $0$.

We now derive expressions for the moments of $\p(\bbeta \mid \by, \sigma^2)$ and the marginal likelihood $p(\by \mid \sigma^2)$ for priors in this family.  As a special case, we easily obtain the posterior mean for $\bbeta$ under a half-Cauchy prior on $\lambda$.

Given $\lambda^2$ and $\sigma^2$, the posterior distribution of $\bbeta$ is multivariate normal, with mean $m$ and variance $V$ given by
$$
m = \left(1 - \frac{1}{1+\lambda^2} \right) \by \quad , \quad V = \left(1 - \frac{1}{1+\lambda^2} \right) \sigma^2 \, .
$$
Define $\kappa = 1/(1+\lambda^2)$.  By Fubini's theorem, the posterior mean and variance of $\bbeta$ are
\begin{eqnarray}
\E(\bbeta \mid \by, \sigma^2) &=& \{ 1 - \E(\kappa \mid \by, \sigma^2) \} \by \label{normalpostmean} \\
\mbox{var}(\bbeta \mid \by, \sigma^2) &=&  \{ 1 - \E(\kappa \mid \by, \sigma^2) \} \sigma^2 \label{normalpostvar} \, ,
\end{eqnarray}
now conditioning only on $\sigma^2$.

It is most convenient to work with $p(\kappa)$ instead:
\begin{equation}
\label{HBdensity1}
p(\kappa) \propto \kappa^{a - 1} \ (1-\kappa)^{b-1} 
\ \left\{ \frac{1}{\tau^2} + \left(1 - \frac{1}{\tau^2} \right) \kappa \right\}^{-1} 
e^{-\kappa s} \, .
\end{equation}
The joint density for $\kappa$ and $\by$ takes the same functional form:
$$
p(y_1 , \ldots, y_p, \kappa) \propto \kappa^{a' - 1} \ (1-\kappa)^{b-1} 
\ \left\{ \frac{1}{\tau^2} + \left(1 - \frac{1}{\tau^2} \right) \kappa \right\}^{-1} 
e^{-\kappa s'} \, ,
$$
with $a' = a + p/2$, and $s' = s + Z / 2\sigma^2$ for $Z = \sum_{i=1}^p y_i^2$.  Hence the posterior for $\lambda^2$ is also a hypergeometric inverted-beta distribution, with parameters $(a', b, \tau^2, s')$.

Next, the moment-generating function of (\ref{HBdensity1}) is easily shown to be
$$
M(t) = e^t \ \frac{\Phi_1(b, 1, a + b, s-t, 1-1/\tau^2)}{ \Phi_1(b, 1, a + b, s, 1-1/\tau^2)} \, .
$$
See, for example, \citet{gordy:1998}.  Expanding $\Phi_1$ as a sum of $\kummer$ functions and using the differentiation rules given in Chapter 15 of \citet{Abra:Steg:1964} yields
\begin{equation}
\label{equation.moments}
\E(\kappa^n \mid \by, \sigma^2) =  \frac{(a')_n}{(a' + b)_n} \frac{\Phi_1(b, 1, a' + b + n, s', 1-1/\tau^2)}{ \Phi_1(b, 1, a' + b, s', 1-1/\tau^2)} \, .
\end{equation}

Combining the above expression with (\ref{normalpostmean}) and (\ref{normalpostvar}) yields the conditional posterior mean and variance for $\bbeta$, given $\by$ and $\sigma^2$.  Similarly, the marginal density $p(\by \mid \sigma^2)$ is a simple expression involving the ratio of prior to posterior normalizing constants:
$$
p(\by \mid \sigma^2) = (2\pi \sigma^2)^{-p/2} \ \exp \left( - \frac{Z}{2\sigma^2} \right) \ \frac{\Be(a', b)}{\Be(a, b)} \ 
\frac{\Phi_1(b, 1, a' + b, s', 1-1/\tau^2)}
{\Phi_1(b, 1, a + b, s, 1-1/\tau^2)} \, .
$$

\begin{figure}
\begin{center}
\includegraphics[width=5.5in]{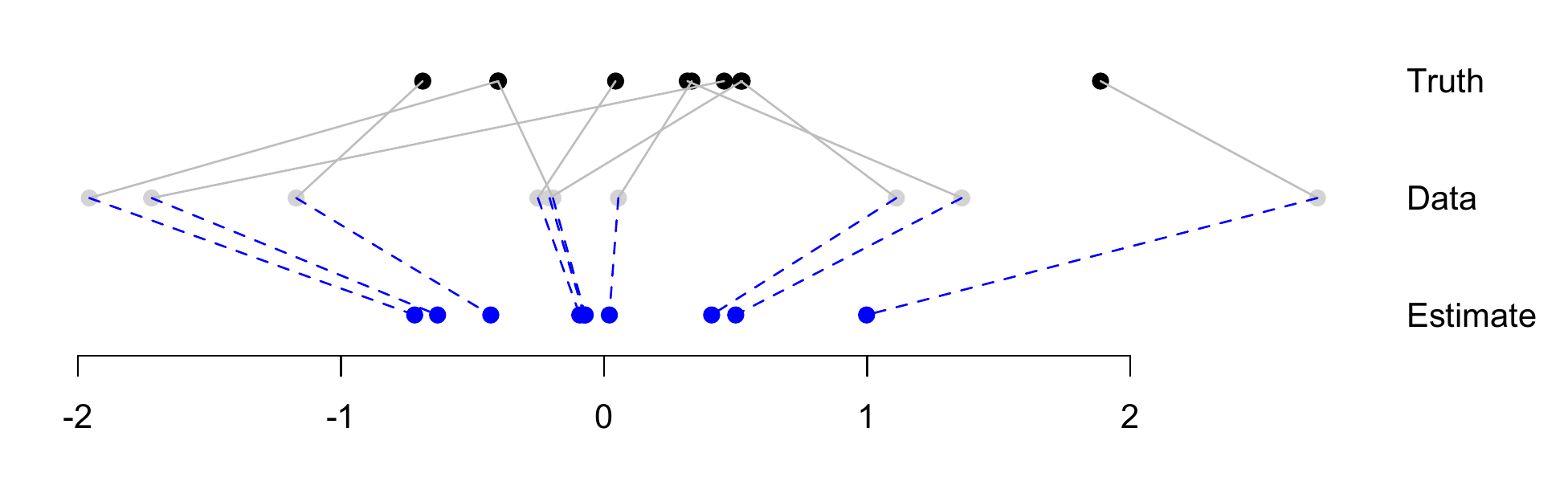}
\caption{\label{fig:shrinkageexample} Ten true means drawn from a standard normal distribution; data from these means under standard normal noise; shrinkage estimates under the half-Cauchy prior for $\lambda$. }
\end{center}
\end{figure}

\section{Classical risk results}

These priors are useful in situations where standard priors like the inverse-gamma or Jeffreys' are inappropriate or ill-behaved.  Non-Bayesians will find them useful for generating easily computable shrinkage estimators that have known risk properties.  Bayesians will find them useful for generating computationally tractable priors for a variance parameter.  We argue that these complementary but overlapping goals can both be satisfied for the special case of the half-Cauchy.  To show this, we first characterize the risk properties of the Bayes estimators that result from the wider family of priors used for a normal mean under a quadratic loss. Our analysis shows that:

\begin{enumerate}
\item The hypergeometric--beta family provides a large class of Bayes estimators that will perform no worse than the MLE in the tails, i.e.~when $\Vert \bbeta \Vert^2$ is large.
\item Major improvements over the James--Stein estimator are possible near the origin.  This can be done in several ways: by choosing $a$ large relative to $b$, by choosing $a$ and $b$ both less than $1$, by choosing $s$ negative, or by choosing $\tau < 1$.  Each of these choices involves a compromise somewhere else in the parameter space.
\end{enumerate}

We now derive expressions for the classical risk, as a function of $\Vert \bbeta \Vert$, for the resulting Bayes estimators under hypergeometric inverted-beta priors.  Assume without loss of generality that $\sigma^2 = 1$, and let $p(\by) = \int p(\by | \bbeta) p(\bbeta) d \bbeta$ denote the marginal density of the data.  Following \citet{stein:1981}, write the the mean-squared error of the posterior mean $\hat{\bbeta}$ as
$$
\E( \Vert \hat{\bbeta} - \bbeta \Vert^2)  = p + \E_{\by} \left ( \Vert g(\by) \Vert^2 + 2 \sum_{i=1}^p \frac{\partial}{\partial y_i} g (\by) \right ) \, ,
$$
where $g(\by) =  \nabla \log p(\by) $. In turn this can be written as
$$
\E( \Vert \hat{\bbeta} - \bbeta \Vert^2) = p + 4 E_{\by \mid \bbeta} \left ( \frac{\nabla^2 \sqrt{p(\by)}}{\sqrt{p(\by)} } \right ) \, .
$$
We now state our main result concerning computation of this quantity.

\begin{proposition}
\label{expression.for.MSE}
Suppose that $\bbeta \sim \N_p(0, \lambda^2 I)$, that $\kappa = 1/(1+\lambda^2)$, and that the prior $\p(\kappa)$ is such that $\lim_{ \kappa \rightarrow 0 , 1} \kappa ( 1 - \kappa ) \p( \kappa ) = 0 $.  Define
$$
m_p ( Z ) = \int_0^1 \kappa^{ \frac{p}{2} } e^{ - \frac{Z}{2} \kappa } \p ( \kappa ) \ \dd \kappa \, 
$$
for $Z = \sum_{i=1}^p y_i^2$.  Then as a function of $\bbeta$, the quadratic risk of the posterior mean under $\p(\kappa)$ is
\begin{equation}
\label{hb.risk.decomposition1}
\E( \Vert \hat{\bbeta} - \bbeta \Vert^2) = p + 2 E_{Z \mid \bbeta} \left\{ Z \frac{  m_{p+4} ( Z ) }{m_{p} (Z)} - p g(Z) - \frac{Z}{2} g(Z)^2 \right\} \, ,
\end{equation}
where
$ g(Z)= E( \kappa \mid Z ) $, and where
\begin{equation}
\label{hb.risk.decomposition2}
Z \frac{  m_{p+4} ( Z ) }{m_{p} (Z)} = ( p+ Z + 4) g(Z) - ( p+2 ) - 
 E_{ \kappa \mid Z } \left\{ 2 \kappa ( 1 - \kappa ) \frac{ \p^{\prime} ( \kappa ) }{ \p ( \kappa ) }
 \right\} \, .
\end{equation}
\end{proposition}

\begin{proof}
See Appendix \ref{appendix.MSEproof}.
\end{proof}

\begin{figure}
\begin{center}
\includegraphics[width=5.5in]{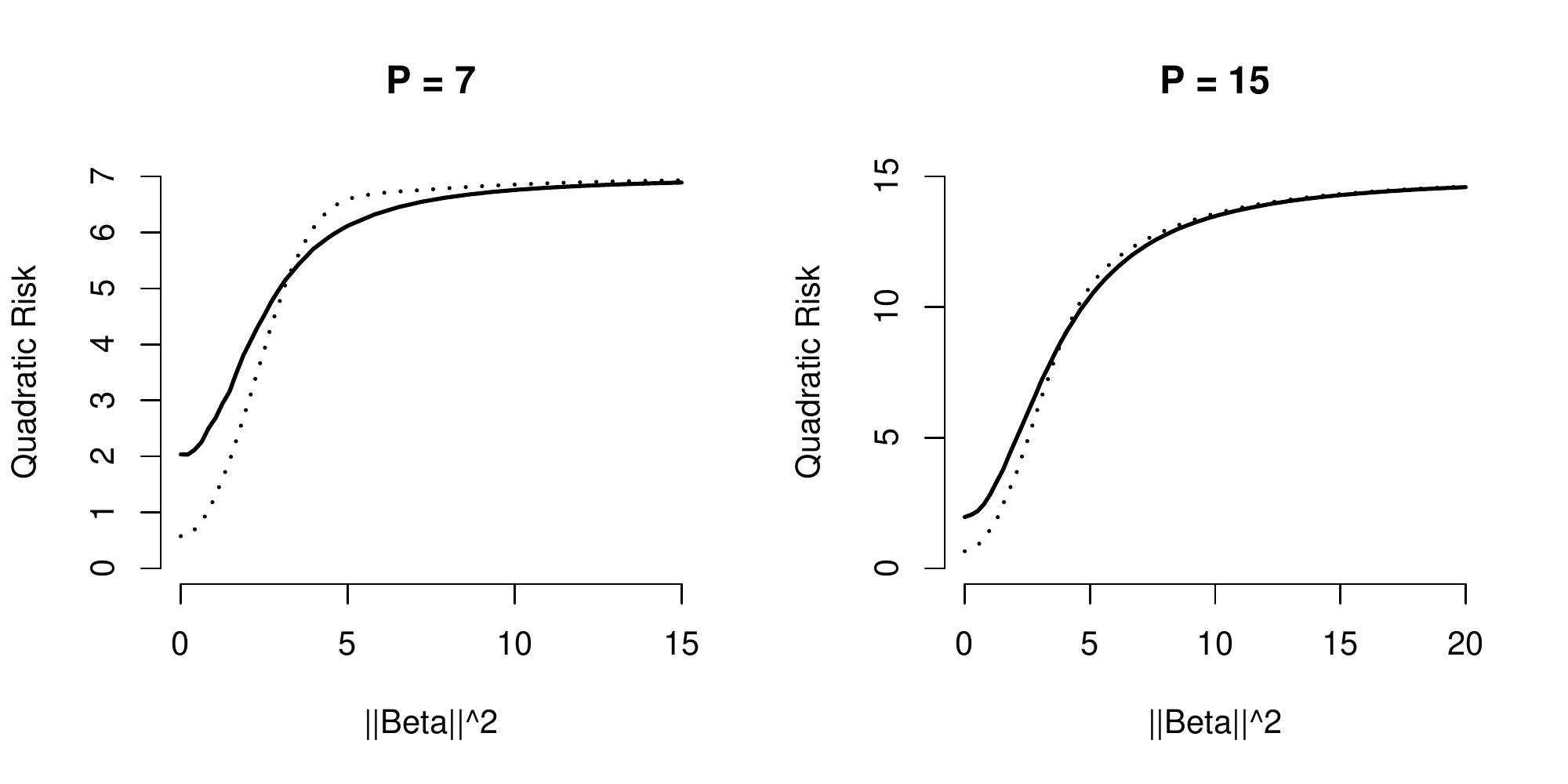}
\caption{\label{horseshoeMSE1.figure} Mean-squared error as a function of $\Vert \bbeta \Vert$ for $p=7$ and $p=15$.  Solid line: James--Stein estimator.  Dotted line: Bayes estimator under a half-Cauchy prior for $\lambda$.}
\end{center}
\end{figure}


Proposition \ref{expression.for.MSE} is useful because it characterizes the risk in terms of two known quantities: the integral $m_p(Z)$, and the posterior expectation $g(Z) = \E(\kappa \mid Z)$.  Using the results of the previous section, these are easily obtained under a hypergeometric inverted-beta prior for $\lambda^2$.  Furthermore, given $\Vert \bbeta \Vert$, $Z = U^2 + V$ in distribution, where
\begin{eqnarray*}
U &\sim& \N \left( \Vert \bbeta \Vert, 1 \right) \\
V &\sim& \chi^2_{p-1} \, .
\end{eqnarray*}
The risk of the Bayes estimator is therefore easy to evaluate as a function of $\Vert \bbeta \Vert^2$.  These expressions can be compared to those of, for example, \citet{george:liang:xu:2006}, who consider Kullback--Leibler predictive risk for similar priors.

Our interest is in the special case $a=b=1/2$, $\tau=1$, and $s=0$, corresponding to a half-Cauchy prior for the global scale $\lambda$.  Figure \ref{horseshoeMSE1.figure} shows the classical risk of the Bayes estimator under this prior for $p=7$ and $p=15$.  The risk of the James-Stein estimator is shown for comparison.  These pictures look similar for other values of $p$, and show overall that the half-Cauchy prior for $\lambda$ leads to a Bayes estimator that is competitive with the James--Stein estimator.

Figure \ref{fig:shrinkageexample} shows a simple example of the posterior mean under the half-Cauchy prior for $\lambda$ when $p=10$, calculated for fixed $\sigma$ using the results of the previous section.  For this particular value of $\bbeta$ the expected squared-error risk of the MLE is $10$, and the expected squared-error risk of the half-Cauchy posterior mean is $8.6$.

A natural question is: of all the hypergeometric inverted-beta priors, why choose the half-Cauchy?  There is no iron-clad reason to do so, of course, and we can imagine many situations where subjective information would support a different choice.  But in examining many other members of the class, we have observed that the half-Cauchy seems to occupy a sensible ``middle ground'' in terms of frequentist risk.  To study this, we are able to appeal to the theory of the previous section.  See, for example, Figure \ref{fig:hypbetaMSE}, which compares several members of the class for the case $p=7$.  Observe that large gains over James--Stein near the origin are possible, but only at the expense of minimaxity.  The half-Cauchy, meanwhile, still improves upon the James--Stein estimator near the origin, but does not sacrifice good risk performance in other parts of the parameter space.  From a purely classical perspective, it looks like a sensible default choice, suitable for repeated general use.

\begin{figure}
\begin{center}
\includegraphics[width=5.2in]{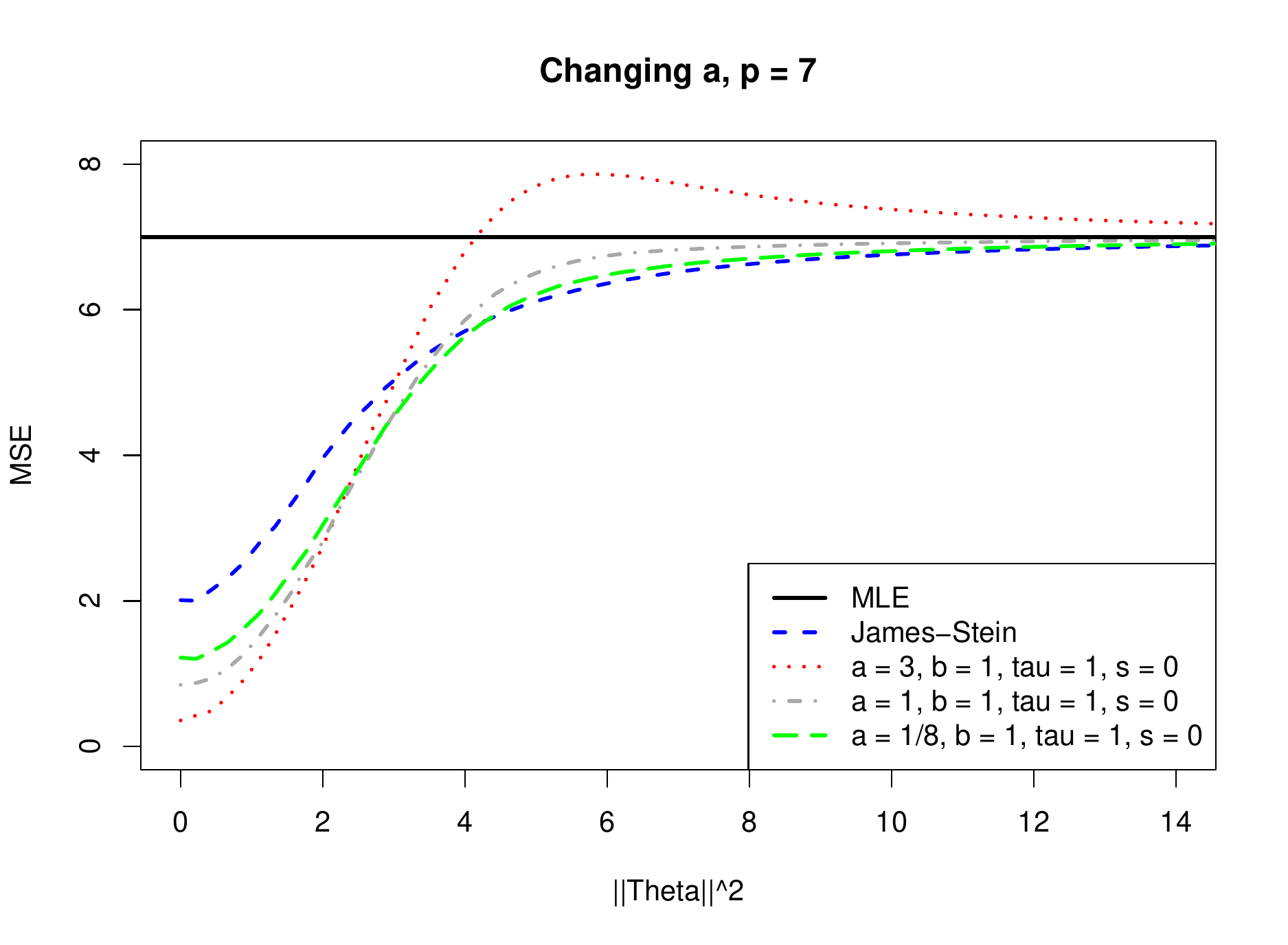} \\
\vspace{1pc}
\includegraphics[width=5.2in]{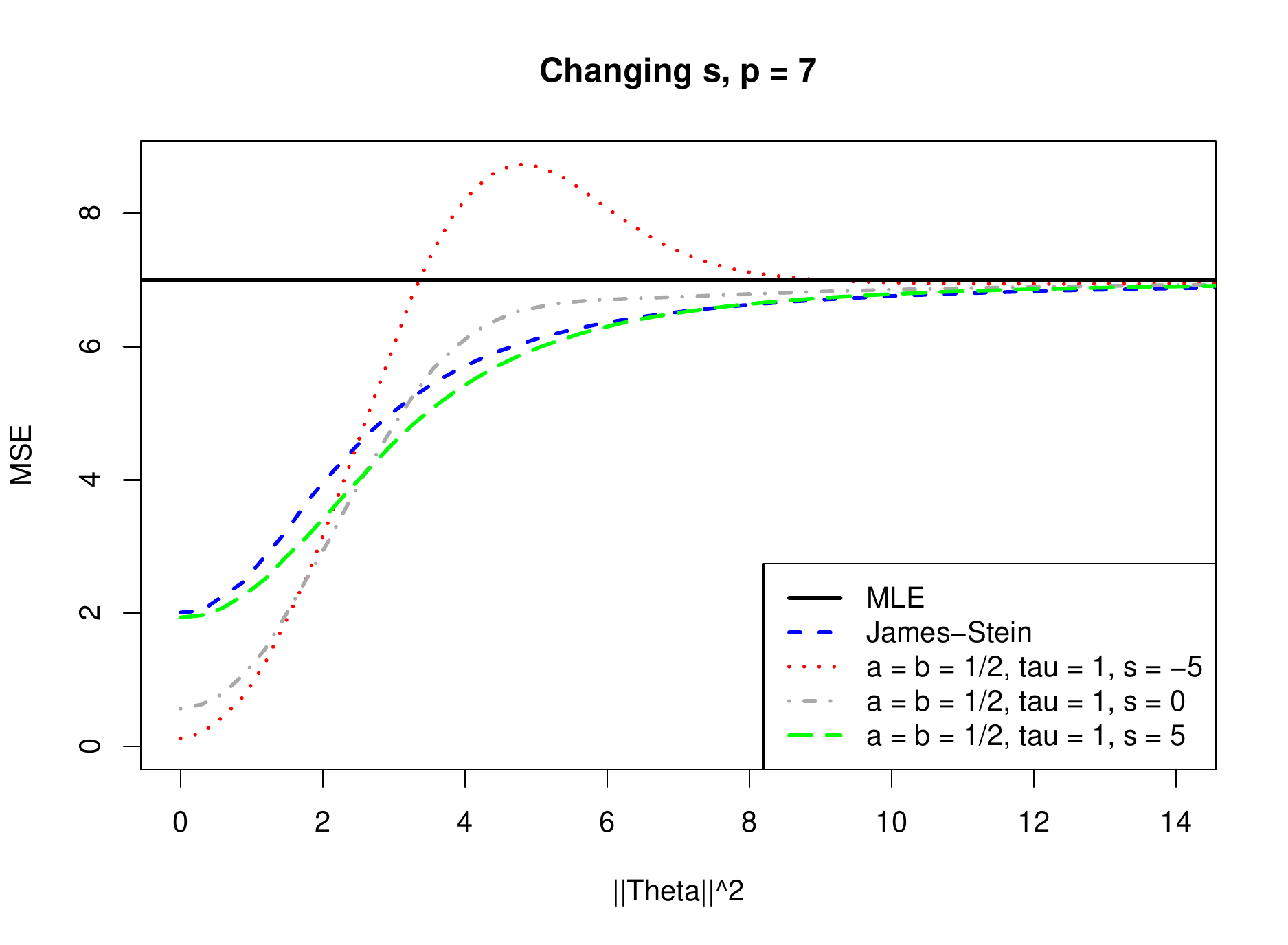}
\caption{\label{fig:hypbetaMSE} Mean-squared error as a function of $\Vert \bbeta \Vert^2$ for $p=7$ and various cases of the hypergeometric inverted-beta hyperparameters.}
\end{center}
\end{figure}

\section{Global scale parameters in local-shrinkage models}

A now-canonical modification of the basic hierarchical model from the introduction involves the use of local shrinkage parameters:
\begin{eqnarray*}
(y_{i} \mid \beta_i, \sigma^2) &\sim& \N(\beta_i, \sigma^2) \\
(\beta_i \mid \lambda^2, u_i^2, \sigma^2) &\sim& \N(0, \lambda^2 \sigma^2 u_i^2) \\
u_i^2 &\sim& f(u_i^2) \\
\lambda^2 &\sim& g(\lambda^2) \, .
\end{eqnarray*}
Mixing over $u_i$ leads to a non-Gaussian marginal for $\beta_i$.  For example, choosing an exponential prior for each $u_i^2$ results in a Laplace (lasso-type) prior.  This class of models provides a Bayesian alternative to penalized-likelihood estimation.  When the underlying vector of means is sparse, these global-local shrinkage models can lead to large improvements in both estimation and prediction compared with pure global shrinkage rules.  There is a large literature on the choice of $p(u_i^2)$, with \citet{polson:scott:2010a} providing a recent review.

As many authors have documented, strong global shrinkage combined with heavy-tailed local shrinkage is why these sparse Bayes estimators work so well at sifting signals from noise.  Intuitively, the idea is that $\lambda$ acts as a global parameter that adapts to the underlying sparsity of the signal.  When few signals are present, it is quite common for the marginal likelihood of $\by$ as a function of $\lambda$ to concentrate near $0$ (``shrink globally''), and for the signals to be flagged via very large values of the local shrinkage parameters $u_i^2$ (''act locally''). Indeed, in some cases the marginal maximum-likelihood solution can be the degenerate $\hat{\lambda} = 0$ \citep[see, for example,][]{tiao:tan:1965}.  

The classical risk results of the previous section no longer apply to a model with these extra local-shrinkage parameters, since the marginal distribution of $\bbeta$, given $\lambda$, is not multivariate normal.  Nonetheless, the case of sparsity serves only to amplify the purely Bayesian argument in favor of the half-Cauchy prior for a global scale parameter---namely, the argument that $p(\lambda \mid \by)$ should not be artificially pulled away from zero by an inverse-gamma prior.

\begin{figure}
\begin{center}
\includegraphics[width=5.5in]{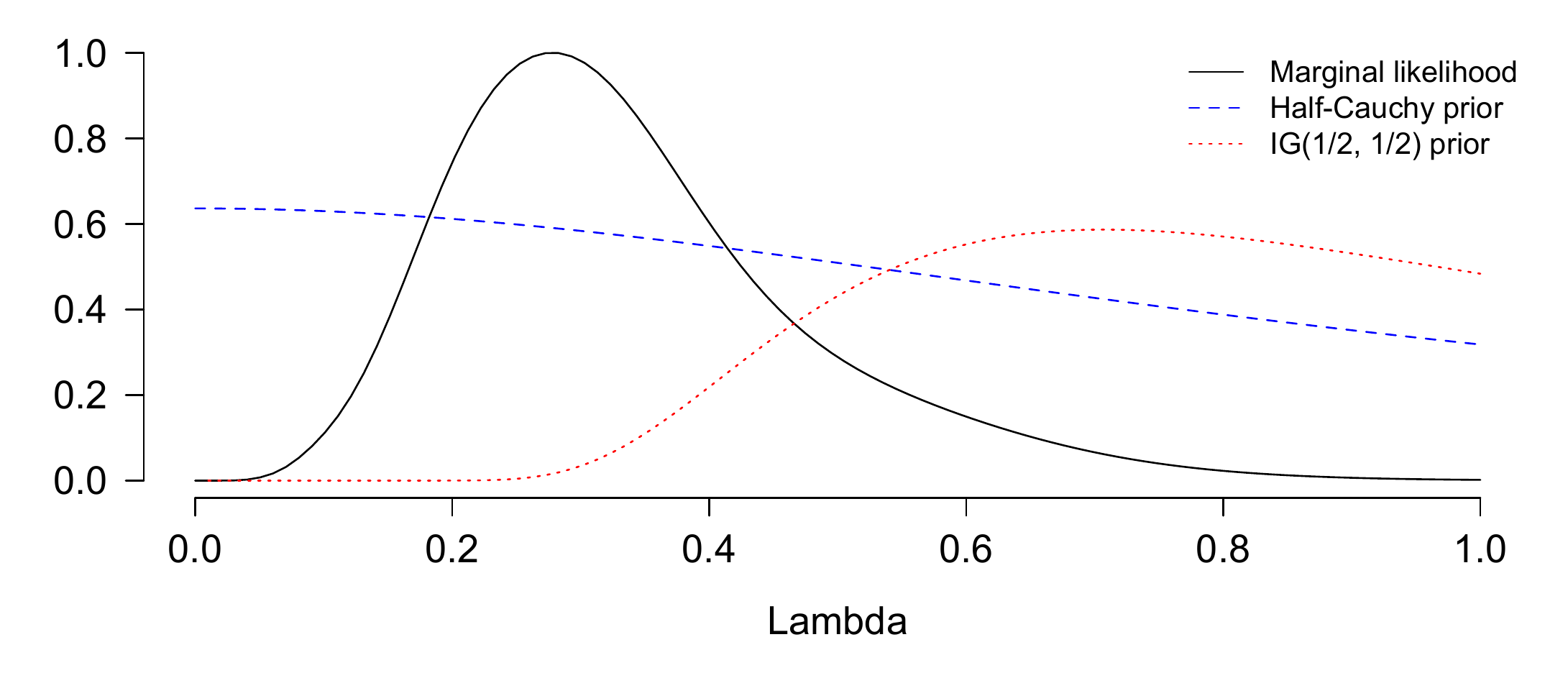}
\caption{\label{fig:MLcomparison} The black line shows the marginal likelihood of the data as a function of $\lambda$ under a horseshoe prior for each $\beta_i$. (The likelihood has been renormalized to obtained a maximum of $1$.)  The blue and red lines show two priors for $\lambda$: the half-Cauchy, and that induced by an inverse-Gamma prior on $\lambda^2$.}
\end{center}
\end{figure}

Figure \ref{fig:MLcomparison} vividly demonstrates this point.  We simulated data from a sparse model where $\bbeta$ contained the entries $(5,4,3,2,1)$ along with $45$ zeroes, and where $y_{ij} \sim \N(0,1)$ for $j=1,2,3$.  We then used Markov Chain Monte Carlo to compute the marginal likelihood of the data as a function of $\lambda$, assuming that each  $\beta_i$ has a horseshoe prior \citep{Carvalho:Polson:Scott:2008a}.  This can be approximated by assuming a flat prior for $\lambda$ (here truncated above at 10), and then computing the conditional likelihood $p(\by \mid \lambda, \sigma, u_1^2, \ldots, u_p^2)$ over a discrete grid of $\lambda$ values at each step of the Markov chain.   The marginal likelihood function can then be approximated as the pointwise average of the conditional likelihood over the samples from the joint posterior.

This marginal likelihood has been renormalized to obtain a maximum of $1$ and then plotted alongside two alternatives: a half-Cauchy prior for $\lambda$, and the prior induced by assuming that $\lambda^2 \sim \IG(1/2,1/2)$.  Under the inverse-gamma prior, there will clearly be an inappropriate biasing of $p(\lambda)$ away from zero, which will negatively affect the ability of the model to handle sparsity efficiently.  For data sets with even more ``noise'' entries in $\by$, the distorting effect of a supposedly ``default'' inverse-gamma prior will be even more pronounced, as the marginal likelihood will favor values of $\tau$ very near zero (along with a small handful of very large $u_i^2$ terms).

\section{Discussion}

On strictly Bayesian grounds, the half-Cauchy is a sensible default prior for scale parameters in hierarchical models: it tends to a constant at $\lambda = 0$; it is quite heavy-tailed; and it leads to simple conjugate MCMC routines, even in more complex settings.  All these desirable features are summarized by \citet{gelman:2006}.  Our results give a quite different, classical justification for this prior in high-dimensional settings: its excellent quadratic risk properties.  The fact that two independent lines of reasoning both lead to the same prior is a strong argument in its favor as a default proper prior for a shared variance component.  We also recommend scaling the $\beta_i$'s by $\sigma$, as reflected in the hierachical model from the introduction.  This is the approach taken by \citet[Section 5.2]{jeffreys1961}, and we cannot improve upon his arguments.

In addition, our hypergeometric inverted-beta class provides a useful generalization of the half-Cauchy prior, in that it allows for greater control over global shrinkage through $\tau$ and $s$.  It leads to a large family of estimators with a wide range of possible behavior, and generalizes the form noted by \citet{maruyama:1999}, which contains the positive-part James--Stein estimator as a limiting, improper case.   Further study of this class may yield interesting frequentist results, quite apart from the Bayesian implications considered here.  The expressions for marginal likelihoods also have connections with recent work on generalized $g$-priors \citep{maruyama:george:2008,polson:scott:2010b}.  Finally, all estimators arise from proper priors on $\lambda^2$, and will therefore be admissible.

There are still many open issues in default Bayes analysis for hierarchical models that are not addressed by our results.  One issue is whether to mix further over the scale in the half-Cauchy prior, $ \lambda \sim C^+ ( 0 , \tau ) $.  One possibility here is simply to let $ \tau \sim C^+ (0,1)$. We then get the following default ``double'' half-Cauchy prior for
$ \lambda$:
$$
p( \lambda ) = \frac{2}{\pi^2} \int_0^\infty \frac{1}{1+ \tau^2} \frac{1}{ \tau( 1 + \frac{\lambda^2}{\tau^2} )} d \tau =
\frac{\ln | \lambda |}{\lambda^2 -1}  \, .
$$
Admittedly, it is difficult to know where to stop in this ``turtles all the way down'' approach to mixing over hyperparameters.  (Why not, for example, mix still further over a scale parameter for $\tau$?)  Even so, this prior has a number of appealing properties.  It is proper, and therefore leads to a proper posterior; it is similar in overall shape to Jeffreys' prior; and it is unbounded at the origin, and will therefore not down-weight the marginal likelihood as much as the half-Cauchy for near-sparse configurations of $\bbeta$.  The implied prior on the shrinkage weight $\kappa$ for the double half-Cauchy is
$$
p( \kappa ) \propto \frac{ \ln \left ( \frac{1-\kappa}{\kappa} \right ) }{ 1 - 2 \kappa } \frac{1}{\sqrt{\kappa(1-\kappa)}} \, .
$$
This is like the horseshoe prior on the shrinkage weight \citep{Carvalho:Polson:Scott:2008a}, but with an extra factor that comes from the fact that one is letting the scale itself be random with a $C^+(0,1)$ prior.

We can also transform to the real line by letting $ \psi = \ln \lambda^2 $.   For the half-Cauchy prior $p(\lambda) \propto 1 /(1 + \lambda^2) $ this transformation leads to
$$
p( \psi ) \propto \frac{ e^{ \frac{\psi}{2} } }{ 1+ e^\psi  } = \left (  e^{ \frac{\psi}{2} } + e^{ - \frac{\psi}{2} } \right )^{-1} 
= sech \left ( \frac{\psi}{2} \right ) \, .
$$
This is the hyperbolic secant distribution, which may provide fertile ground for further generalizations or arguments involving sensible choices for a default prior.

A more difficult issue concerns the prior scaling for $\lambda$ in the presence of unbalanced designs---that is, when $y_{ij} \sim \N(\beta_i, \sigma^2)$ for $j = 1, \ldots, n_i$, and the $n_i$'s are not necessary equal.  In this case most formal non-informative priors for $\lambda$ (e.g.~the reference prior for a particular parameter ordering) involve complicated functions of the $n_i$'s \citep[see, e.g.][]{yang:berger:1997}.  These expressions emerge from a particular mathematical formalism that, in turn, embodies a particular operational definition of ``non-informative.''

We have focused on default priors that occupy a middle ground between formal non-informative analysis and pure subjective Bayes.  This is clearly an important situation for the many practicing Bayesians who do not wish to use noninformative priors, whether for practical, mathematical, or philosophical reasons.  An example of a situation in which formal noninformative priors for $\lambda$ should not be used on mathematical grounds is when $\bbeta$ is expected to be sparse; see \cite{scottberger06} for a discussion of this issue in the context of multiple-testing.  It is by no means obvious how, or even whether, the $n_i$'s should play a role in scaling $\lambda$ within this (admittedly ill-defined) paradigm of ``default'' Bayes.

Finally, another open issue is the specification of default priors for scale parameters in non-Gaussian models.   For example, in logistic regression, the likelihood is highly sensitive to large values of the underyling linear predictor.  It is therefore not clear whether something so heavy-tailed as the half-Cauchy is an appropriate prior for the global scale term for logistic regression coefficients.  All of these issues merit further research.

%
%

\appendix

\section{Details for computing moments and marginals}
\label{hyperg.integral.details.appendix}

The normalizing constant in (\ref{normalizing.constant.phi1}) is
\begin{equation}
C = \int_0^{1}  \kappa^{\alpha - 1} \ (1 - \kappa)^{\beta - 1} \ \left\{ \frac{1}{\tau^2} + \left(1 - \frac{1}{\tau^2} \right) \kappa \right\}^{-1} \exp(-s \kappa)  \ \dd \kappa \, .
\end{equation}
Let $\eta = 1-\kappa$.  Using the identity that $e^x = \sum_{m=0}^{\infty} x^m / m!$, we obtain
$$
C = e^{-s} \sum_{m=0}^{\infty} \left[ \frac{s^m}{m!} \int_0^1 \eta^{\beta + m -1} (1-\eta)^{\alpha - 1} \{1 - (1-1/\tau^2)\eta\}^{-1} \ \dd \eta \right] \, .
$$
Using properties of the hypergeometric function $\phantom{}_2 F_1$ \citep[\S15.1.1 and \S15.3.1]{Abra:Steg:1964}, this becomes, after some straightforward algebra,
\begin{equation}
\label{doubleseries.expansion}
C = e^{-s} \ \Be(\alpha, \beta) \ \sum_{m=0}^{\infty} \sum_{n=0}^{\infty} \frac{(\beta)_{m+n} }
{(\alpha + \beta)_{m+n} \ m! \ n!} 
\ s^m \ (1-1/\tau^2)^m \, ,
\end{equation}
where $(a)_n$ is the rising factorial.  Appendix C of \citet{gordy:1998} proves that,  for all $\alpha > 0$, $\beta > 0$, and $1/\tau^2 > 0$, the nested series in (\ref{doubleseries.expansion}) converges to a positive real number, yielding
\begin{equation}
C = e^{-s} \ \Be(\alpha, \beta) \ \Phi_1(\beta, 1, \alpha + \beta, s, 1-1/\tau^2) \, ,
\end{equation}
where $\Phi_1$ is the degenerate hypergeometric function of two variables \citep[9.261]{gradshteyn:ryzhik:1965}.

The $\Phi_1$ function can be written as a double hypergeometric series,
\begin{equation}
\Phi_1(\alpha, \beta; \gamma; x,y) = \sum_{m=0}^{\infty} \sum_{n=0}^{\infty}
\frac
{(\alpha)_{m+n} (\beta)_{n} }
{(\gamma)_{m+n} m! n!}
\ y^n \ x^m \, ,
\end{equation}
where $(c)_n$ is the rising factorial.  We use three different representations of $\Phi_1(\alpha, \beta, \gamma, x, y)$ for handling different combinations of arguments, all from \citet{gordy:1998}.  When $0 \leq y < 1$ and $x \geq 0$,
\begin{equation}
\Phi_1(\alpha, \beta, \gamma, x, y) =
\sum_{n=0}^{\infty} \frac{(\alpha)_n } {(\gamma)_n} \frac{x^n}{n!} \ \kummertwo(\beta,\alpha+n; \gamma+n; y) \label{phi1-1} \, .
\end{equation}
When $0 \leq y < 1$ and $x < 0$,
\begin{equation}
\Phi_1(\alpha, \beta, \gamma, x, y) = e^x \ \sum_{n=0}^{\infty} \frac{(\gamma - \alpha)_n } {(\gamma)_n} \frac{(-x)^n}{n!} \ \kummertwo(\beta,\alpha; \gamma+n; y)  \label{phi1-2} \, .
\end{equation}
Finally, when $y<0$,
\begin{equation}
\Phi_1(\alpha, \beta, \gamma, x, y)  = e^x \ (1-y)^{-\beta} \ \Phi_1(\tilde{\alpha}, \beta, \gamma, -x, \tilde{y}) \label{phi1-3} \, ,
\end{equation}
where $\tilde{\alpha} = \gamma - \alpha$ and $\tilde{y} = y/(y-1)$.  Then either (\ref{phi1-1}) or (\ref{phi1-2}) may be used to evaluate the righthand side of (\ref{phi1-3}), depending on the sign of $x$.


\section{Proof of Proposition \ref{expression.for.MSE}}
\label{appendix.MSEproof}

\begin{proof}
Begin with Stein's decomposition of risk.  Following Equation (10) of \citet{fourdrinier:etal:1998}, we have
$$
\Vert \nabla m( \by) \Vert = \Vert \by \Vert 
 \int_0^1 \kappa^{ \frac{p}{2} +1} \p(\kappa) e^{ - \frac{Z}{2} \kappa } \ \dd \kappa 
$$
The score can be written as
$$
\frac{ \Vert \nabla m( \by) \Vert}{p(\by)} = \Vert \by \Vert \frac{m_{p+2} ( \Vert \by \Vert )}{m_{p} ( \Vert \by \Vert )} 
 = \Vert \by \Vert \ \E( \kappa \mid Z ) \, ,
$$
and the Laplacian term is $\Delta m( \by) = \int_0^1 \left ( Z \kappa - p \right ) \kappa^{ \frac{p}{2} +1 } \p(\kappa) e^{- \frac{Z}{2} \kappa } \ \dd \kappa$.  Combining these terms, we have,
\begin{eqnarray*}
 \frac{ \Delta m( \by) }{ p(\by) } &=& \frac{\int_0^1 \left ( Z \kappa - p \right ) \kappa^{ \frac{p}{2} +1 } 
\p(\kappa) e^{ - \frac{Z}{2} \kappa } \ \dd \kappa}
{\int_0^1 \kappa^{ \frac{p}{2}} \p(\kappa) e^{ - \frac{Z}{2} \kappa } \ \dd \kappa } \\
&=& Z \frac{m_{p+4}(Z)}{m_p(Z)} - p \frac{m_{p+2} (Z)}{m_p(Z)} \, .
\end{eqnarray*}

The risk term $ \Delta \sqrt{p(\by)} / \sqrt{p(\by)} $ is then computed using the identity
$$
 \frac{ \nabla^2 \sqrt{p(\by)}}{ \sqrt{p(\by)} } =
 \frac{1}{2} \left[ \frac{ \Delta p(\by) }{ p(\by) } - \frac{1}{2} \left\{ \frac{ \Vert \nabla p(\by) \Vert}{p(\by)} \right\}^2 \right] \, ,
$$
which reduces to
$$
\frac{1}{2} \left\{ Z \frac{  m_{p+4} ( Z ) }{m_{p} (Z)} - p g(Z) - \frac{Z}{2} g(Z)^2 \right\}
$$
for $ g(Z) = E( \kappa \mid Z )$.

Secondly, note that
\begin{equation}
Z \{ m_{p+2} (Z ) - m_{p+4} (Z ) \} =
  2 \int_0^1 \kappa^{\frac{p}{2}+1}  ( 1 - \kappa ) \p( \kappa ) d \left ( - e^{ - \frac{Z}{2} \kappa } \right ) \, .
\end{equation}
Therefore,
$$
Z \left \{ \frac{m_{p+2} (Z )}{m_p(Z)} - \frac{m_{p+4} (Z )}{m_p(Z)} \right \} =
\int_0^1 \left\{ ( p+2) ( 1 - \kappa ) - 2 \kappa + 2 \kappa ( 1 - \kappa ) \frac{ \p^{\prime} (\kappa) }{ \p(\kappa) }
\right \} \frac{\kappa^{ \frac{p}{2}} e^{- \frac{Z}{2} \kappa} \p( \kappa )}{m_p(Z)} \ \dd \kappa
$$
Then under the assumption that $ \lim_{ \kappa \rightarrow 0 , 1} \kappa ( 1 - \kappa ) \p( \kappa ) = 0 $, integration by parts gives (\ref{hb.risk.decomposition2}).  Hence
$$
\E( \Vert \hat{\bbeta} - \bbeta \Vert^2)  = p + 2 \E_{Z \mid \theta} \left[ (Z+4) g(Z) - (p+2)  - \frac{Z}{2} g(Z)^2  - 
  E_{ \kappa | Z } \left\{ 2 \kappa ( 1 - \kappa ) \frac{ \p^{\prime} ( \kappa ) }{ \p ( \kappa ) }
\right\} \right] \, .
$$
\end{proof}

\singlespace

\begin{small}

\bibliographystyle{abbrvnat}
\bibliography{masterbib}

\begin{thebibliography}{20}
\providecommand{\natexlab}[1]{#1}
\providecommand{\url}[1]{\texttt{#1}}
\expandafter\ifx\csname urlstyle\endcsname\relax
  \providecommand{\doi}[1]{doi: #1}\else
  \providecommand{\doi}{doi: \begingroup \urlstyle{rm}\Url}\fi

\bibitem[Abramowitz and Stegun(1964)]{Abra:Steg:1964}
M.~Abramowitz and I.~A. Stegun, editors.
\newblock \emph{Handbook of Mathematical Functions With Formulas, Graphs, and
  Mathematical Tables}, volume~55 of \emph{Applied Mathematics Series}.
\newblock National Bureau of Standards, Washington, DC, 1964.
\newblock Reprinted in paperback by Dover (1974); on-line at
  \url{http://www.math.sfu.ca/$\sim$cbm/aands/}.

\bibitem[Berger(1980)]{BergerAnnals1980}
J.~O. Berger.
\newblock A robust generalized {B}ayes estimator and confidence region for a
  multivariate normal mean.
\newblock \emph{The Annals of Statistics}, 8\penalty0 (4):\penalty0 716--761,
  1980.

\bibitem[Carvalho et~al.(2010)Carvalho, Polson, and
  Scott]{Carvalho:Polson:Scott:2008a}
C.~M. Carvalho, N.~G. Polson, and J.~G. Scott.
\newblock The horseshoe estimator for sparse signals.
\newblock \emph{Biometrika}, 97\penalty0 (2):\penalty0 465--80, 2010.

\bibitem[Fourdrinier et~al.(1998)Fourdrinier, Strawderman, and
  Wells]{fourdrinier:etal:1998}
D.~Fourdrinier, W.~Strawderman, and M.~T. Wells.
\newblock On the construction of {B}ayes minimax estimators.
\newblock \emph{The Annals of Statistics}, 26\penalty0 (2):\penalty0 660--71,
  1998.

\bibitem[Gelman(2006)]{gelman:2006}
A.~Gelman.
\newblock Prior distributions for variance parameters in hierarchical models.
\newblock \emph{Bayesian Anal.}, 1\penalty0 (3):\penalty0 515--33, 2006.

\bibitem[George et~al.(2006)George, Liang, and Xu]{george:liang:xu:2006}
E.~I. George, F.~Liang, and X.~Xu.
\newblock Improved minimax predictive densities under {K}ullback-{L}eibler
  loss.
\newblock \emph{The Annals of Statistics}, 34\penalty0 (1):\penalty0 78--91,
  2006.

\bibitem[Gordy(1998)]{gordy:1998}
M.~B. Gordy.
\newblock A generalization of generalized beta distributions.
\newblock Finance and Economics Discussion Series 1998-18, Board of Governors
  of the Federal Reserve System (U.S.), 1998.

\bibitem[Gradshteyn and Ryzhik(1965)]{gradshteyn:ryzhik:1965}
I.~Gradshteyn and I.~Ryzhik.
\newblock \emph{Table of Integrals, Series, and Products}.
\newblock Academic Press, 1965.

\bibitem[Griffin and Brown(2005)]{griffin:brown:2005}
J.~Griffin and P.~Brown.
\newblock Alternative prior distributions for variable selection with very many
  more variables than observations.
\newblock Technical report, University of Warwick, 2005.

\bibitem[Jeffreys(1961)]{jeffreys1961}
H.~Jeffreys.
\newblock \emph{Theory of Probability}.
\newblock Oxford University Press, 3rd edition, 1961.

\bibitem[Maruyama(1999)]{maruyama:1999}
Y.~Maruyama.
\newblock Improving on the {J}ames--{S}tein estimator.
\newblock \emph{Statistics and Decisions}, 14:\penalty0 137--40, 1999.

\bibitem[Maruyama and George(2010)]{maruyama:george:2008}
Y.~Maruyama and E.~I. George.
\newblock $g$bf: A fully {B}ayes factor with a generalized g-prior.
\newblock Technical report, University of Tokyo, arXiv:0801.4410v2, 2010.

\bibitem[Morris and Tang(2011)]{morris:tang:2011}
C.~Morris and R.~Tang.
\newblock Estimating random effects via adjustment for density maximization.
\newblock \emph{Statistical Science}, 26\penalty0 (2):\penalty0 271--87, 2011.

\bibitem[Polson and Scott(2011{\natexlab{a}})]{polson:scott:2010a}
N.~G. Polson and J.~G. Scott.
\newblock Shrink globally, act locally: sparse {B}ayesian regularization and
  prediction.
\newblock In \emph{Proceedings of the 9th Valencia World Meeting on Bayesian
  Statistics}. Oxford Univeristy Press, 2011{\natexlab{a}}.

\bibitem[Polson and Scott(2011{\natexlab{b}})]{polson:scott:2010b}
N.~G. Polson and J.~G. Scott.
\newblock Local shrinkage rules, {L}\'evy processes, and regularized
  regression.
\newblock \emph{Journal of the Royal Statistical Society (Series B)}, (to
  appear), 2011{\natexlab{b}}.

\bibitem[Scott and Berger(2006)]{scottberger06}
J.~G. Scott and J.~O. Berger.
\newblock An exploration of aspects of {B}ayesian multiple testing.
\newblock \emph{Journal of Statistical Planning and Inference}, 136\penalty0
  (7):\penalty0 2144--2162, 2006.

\bibitem[Stein(1981)]{stein:1981}
C.~Stein.
\newblock Estimation of the mean of a multivariate normal distribution.
\newblock \emph{The Annals of Statistics}, 9:\penalty0 1135--51, 1981.

\bibitem[Strawderman(1971)]{strawderman:1971}
W.~Strawderman.
\newblock Proper {B}ayes minimax estimators of the multivariate normal mean.
\newblock \emph{The Annals of Statistics}, 42:\penalty0 385--8, 1971.

\bibitem[Tiao and Tan(1965)]{tiao:tan:1965}
G.~C. Tiao and W.~Tan.
\newblock {B}ayesian analysis of random-effect models in the analysis of
  variance. i. {P}osterior distribution of variance components.
\newblock \emph{Biometrika}, 51:\penalty0 37--53, 1965.

\bibitem[Yang and Berger(1997)]{yang:berger:1997}
R.~Yang and J.~O. Berger.
\newblock A catalog of noninformative priors.
\newblock Technical Report~42, Duke University Department of Statistical
  Science, 1997.

\end{thebibliography}

\end{small}

\end{document}